\documentclass[12pt,english,letterpaper]{article}

\usepackage[english]{babel}
\usepackage{amsmath,amsthm,amssymb,amsfonts,amscd,amstext}
\usepackage{mathrsfs}
\usepackage{bm}
\usepackage{xcolor}
\usepackage{graphicx}
\usepackage{enumerate}
\usepackage[title]{appendix}
\usepackage{hyperref}
\hypersetup{
    colorlinks=true,
    citecolor=blue,
}

\allowdisplaybreaks

\newtheorem{theorem}{Theorem}
\newtheorem{corollary}[theorem]{Corollary}
\newtheorem{definition}[theorem]{Definition}
\newtheorem{example}[theorem]{Example}
\newtheorem{lemma}[theorem]{Lemma}

\newtheorem{proposition}[theorem]{Proposition}
\newtheorem{remark}[theorem]{Remark}

\newcommand{\norm}[1]{\lVert#1\rVert}
\newcommand{\abs}[1]{\left|#1\right|}

\renewcommand{\natural}{\mathbb{N}}
\newcommand{\real}{\mathbb{R}}
\newcommand{\complex}{\mathbb{C}}
\newcommand{\koopman}{\mathcal{K}}
\newcommand{\ruelle}{\mathcal{L}}
\newcommand{\mult}{M}
\newcommand{\wass}{W}
\newcommand{\rel}{\mathcal{R}}
\newcommand{\dirac}{\mathcal{D}}

\newcommand{\cont}{\mathcal{C}}
\newcommand{\prob}{\mathcal{P}}
\newcommand{\bound}{\mathcal{B}}
\newcommand{\states}{\mathcal{S}}
\newcommand{\lp}[1]{L^{#1}}

\newcommand{\id}{\operatorname{Id}}

\newcommand{\defn}{:=}
\renewcommand{\d}{\operatorname{d} \!}

\begin{document}

\title{The Dirac operator for the Ruelle-Koopman pair on \texorpdfstring{$\lp{p}$}{Lp}-spaces: an interplay between Connes distance and symbolic dynamics}

\author{William M. M. Braucks and Artur O. Lopes}

\maketitle

\begin{abstract}
{\scriptsize 
  Denote by $\bm{\mu}$ the maximal entropy measure for the shift \(\sigma\) acting on $\Omega = \{0, 1\}^\mathbb{N}$, by
  $\ruelle$ the associated Ruelle operator and by $\koopman = \ruelle^{\dagger}$ the Koopman operator, both acting on
  $\lp{2}(\bm{\mu})$. Using a diagonal representation $\pi$, the Ruelle-Koopman pair can be used for defining
  a dynamical Dirac operator $\mathcal{D},$ as in \cite{BL}. $\mathcal{D}$ plays the role of a derivative. In
  \cite{lpspec}, the notion of a spectral triple was   generalized to \(\lp{p}\)-operator algebras; in consonance, here, we
  generalize results for $\mathcal{D}$  to results for a Dirac operator $\mathcal{D}_p$ , and the associated Connes distance $d_p$, 
  to this new \(\lp{p}\) context, \(p \geq 1\). Given the states  $\eta, \xi$:
  \vspace{-6pt}
  \begin{center} 
     $d_{p}(\eta, \xi) \defn \sup \{ \,|\eta(a) - \xi(a) |
    \text{, where } a \in \mathcal{A} \text{ and } 
    \norm{\left[\mathcal{D}_p,\pi(a)\right]} \leq 1\}$.
 \end{center}
  \vspace{-6pt}
 \noindent  
   In the setting of operator algebras a function $f\in \mathcal{A}=\mathcal{C}(\Omega)$ is represented by the operator  $\mult_f$, where
  $\mult_f(g)=f\, g$. The operator $M_f$ acts on $L^p (\mu).$ We explore the  relationship of $\mathcal{D}_p$ with dynamics,
  in particular with $f \circ \sigma - f$, the discrete-time derivative of a continuous $f:\Omega \to \mathbb{R}$. Take
  $p,p^{\prime}>0$ satisfying $\frac{1}{p} + \frac{1}{ p^{\prime}}=1$. We show for any continuous function $f$:
  \vspace{-8pt}
 \begin{center} 
 $\norm{\left[ \dirac_p, \pi(\mult_f) \right]} = | \sqrt[\lambda]{\ruelle \abs{f \circ \sigma - f}^{\lambda}} |_{\infty}$, 
    where $\lambda = \max\{p, p^\prime\}$.
  \end{center}
    \vspace{-8pt}
 \noindent 
  Denote by $\prob(\sigma)$ the set of $\sigma$-in\-va\-ri\-ant probabilities; then we get:
  \vspace{-8pt}
\begin{center} 
     $\norm{\left[ \dirac_p , \pi(M(f)) \right]}
    \geq \sqrt[\lambda]{2} \sup_{\mu \in \prob(\sigma)}
    \exp ( \int \log \abs{f \circ \sigma - f} \d \mu + \frac{h_{\mu}(\sigma)}{\lambda} )$. 
  \end{center}
 \vspace{-5pt}
\noindent 
  When \(p = 2\), the equality holds. We analyze the connection of $d_p$ with transport theory. Let \(\mu, \nu\) 
  be probabilities on $\Omega$, \(d^{\infty}\) a certain metric on $\Omega$ and $\wass_{d^{\infty}}$ its
 Wasserstein distance:
  \vspace{-15pt}
  \begin{center}
  $\wass_{d^{\infty}}(\mu, \nu) \leq d_p(\mu, \nu) \leq \sqrt[\lambda]{2} \wass_{d^{\infty}}(\mu, \nu)$.
 \end{center}
 \vspace{-11pt}
  
  \smallskip

  \noindent 
  Moreover, \(d_1 (\mu, \nu)= d_{\infty}(\mu, \nu) = W_{d^{\infty}}(\mu, \nu)\). \(d^{\infty}\) is not
  compatible with the usual   metric. Furthermore, we  show $\norm{\left[ \mathcal{D}_p, \pi(\koopman^{n} \mathcal{L}^{n})]\right]}=1$ 
  for all \(n \geq 1\).  We also prove
 a formula analogous to the Kantorovich duality formula for minimizing the cost of tensor products.}
 \end{abstract}

  \section{Introduction}

  Our main goal here is to introduce a dynamical Dirac operator $\mathcal{D}_p$, $p \geq 1$ (associated to the
  Ruelle-Koopman pair) and to study its action on  {\color{blue}\bf $\lp{p}$}-spaces of functions $f$ defined on the symbolic space
  $\Omega= \{0, 1\}^{\natural}$. The $L^p(\bm{\mu})$ space concerns the maximal entropy measure $\bm{\mu}$. In the case
  of the space $L^2( \bm{\mu})$, the Dirac operator $\mathcal{D}_{2}$ plays the role of a derivative acting on
  self-adjoint operators.

  We will  present results for the associated Connes distance and the corresponding spectral triple. We will also present
  some explicit estimates describing the interplay of all these concepts with the dynamics of the symbolic space. Given
  a continuous function $f:\Omega \to \mathbb{R}$, the function $f \circ \sigma - f$ denotes the discrete-time
  derivative of $f$, where $\sigma$ is the shift acting on $\Omega$. The discrete-time derivative plays an important
  role in our reasoning.

  The Koopman operator $\mathcal{K}$ and the Ruelle operator $\mathcal{L}$ are defined  in \eqref{yet}.  Also, $d_p$
   will denote the Connes distance for the dynamical Dirac operator $\mathcal{D}_p$, $p \geq 1$ (see   expressions
  \eqref{jiu} and \eqref{connesdefu}). $\prob(\Omega)$ denotes the set of probabilities on $\Omega$, and $\prob(\sigma)$
  the set of $\sigma$-in\-va\-ri\-ant probabilities.

  Some of the results we will present here   generalize our previous work \cite{BL} (see also \cite{rkboson}), which considered the case $p=2$. Furthermore, we will  estimate the Connes distance from $\mathcal{D}_p$  in terms of the Wasserstein distance 
  from a cost function described in Section \ref{gen} .

   In another previous work \cite{diss} (a Master dissertation), a ``noncommutative generalization'' of the optimal transport problem was
  considered and it was shown  to satisfy a duality formula analogous to the Kantorovich duality formula (see the
  \nameref{kntrvch}). The dual form of this noncommutative optimal transport problem is of interest to us in the present
  context because it can be related to the Connes distance. More precisely, the Connes distance between two states is
  bounded above by the noncommutative optimal transport cost between the same states (for a given ``noncommutative cost
  function''). This is a corollary of the fact that the Connes distance is bounded above by the optimal transport cost
  (for a given cost function) and that the noncommutative optimal transport problem generalizes the optimal transport
  problem  (see \cite{connes-opt}).

  There are circumstances under which the Connes distance and the optimal transport problem coincide. For instance, in
  this prototypical example of a spectral triple, where the commutative C*-algebra of continuous complex-valued
  functions of a compact manifold acts via multiplication operators on the Hilbert space of square-integrable
  differential forms; and the Dirac operator is the Hodge (or signature) operator \cite{connes-opt}. In this example,
  the states are the Borel probability measures, and the transport cost function under consideration is the manifold's
  metric distance.

  In general, it is known that even for commutative and finite-dimensional C*-algebras \(\mathcal{A}\), the Connes
  distance between two probability vectors may not coincide with any\footnote{This follows from the observation that in
  \cite{connes-opt}, the set of admissible for the Connes distance is an ellipsoid, while the set of admissibles for
  the optimal transport problem is a rectangular prism (whichever the cost).} (as in, for any cost function) optimal
  transport problem between the same probability vectors (see \cite{connes-opt}). 

  In a more general setting, recall that \(\mathcal{A}\) is (an unital, separable, and) commutative (C*-algebra)
  precisely when it is (isometrically isomorphic to) the C*-algebra of continuous complex-valued functions of a given
  compact metric space\footnote{Such a metric space can then be taken to be the so-called spectrum of \(\mathcal{A}\).
  A more detailed description of the spectrum can be found in \cite{cdixmier}.}. The states of \(\mathcal{A}\) can then
  be regarded as (Borel) probability measures on such space. One important issue is: given \(\mu, \nu \in
  \prob(\Omega)\), to characterize when $d_p(\mu,\nu)$ is finite (see Corollary \ref{oprt}). 

  As a means of probing the question of how the Connes distance relates to the optimal transport problem and our version
  of noncommutative optimal transport, we are going to carry out explicit computations in a specifically chosen example.
  This example concerns a form of special variation of the classical spectral triple definition  because of two
  exceptions; the first one being our Dirac operator $\mathcal{D}_p$ is bounded, and therefore does not have compact
  resolvent. However, this will not be an issue, since we are interested mainly in a metric question, for which such
  a hypothesis is of no pertinence. Furthermore, the example will consist of a dynamically defined distance between
  probability measures, which is of interest \textit{per se}.

  Another interesting question regarding the Connes distance is how does it change with respect to the parameter \(p\)
  in the context of \(\lp{p}\)-operator algebras. Hence, the other exception is we will also consider our algebra as an
  \(\lp{p}\)-operator algebra. In doing this, we hope that our explicit computations come to offer a little bit of
  insight into such a question.
  
  Results relating Spectral Triples and Ergodic Theory can be found in  \cite{Kesse}, \cite{Sharp1}, \cite{Sharp2}, \cite{BL} and  \cite{CHLS}.
  
  \section{Notation}

  Let \(\Omega = \left\{0, 1\right\}^{\natural}\) be equipped with the product topology and the corresponding Borel
  \(\sigma\)-algebra. Typical sequences are written \(x, y \in \Omega\). Consider the action of the shift map \(\sigma
  : \Omega \longrightarrow \Omega\) defined by \(x = (x_{n})_{n \in \natural} \longmapsto \sigma(x) = (x_{n+1})_{n \in
  \natural} \text{,}\) and denote by \(\bm{\mu}\) the maximal entropy measure. We write \(\bm{\mu} \sqcup \bm{\mu}\) for
  the measure over the disjoint space \(\Omega \sqcup \Omega\) which restricts to either component as \(\bm{\mu}\).

   Let \(\lp{p} = \lp{p}(\mu)\)   be the Banach space of \(p\)-integrable complex-valued functions   $g:\Omega \to \mathbb{C}$, 
  with respect to the maximal entropy measure $\mu$.

  In   our setting a continuous function $f:\Omega \to \mathbb{C}$ is represented by the operator   $\mult_f$, given by 
   f $g \mapsto \mult_f(g) := f g$. We will introduce a Dirac operator $\mathcal{D}_p$, $p\geq 1$, acting on   pairs
  of \(p\)-integrable functions (see Section \ref{dir}). Our initial focus will  be on the commutator of $\mathcal{D}_p$
    with operators of the form   $\pi(M_f)$. That is, on:
  \begin{equation*} 
    [ \mathcal{D}_p, \pi(M_f)]  {\color{blue}\bf \text{,} 
  }\end{equation*} 
  where $[\,,\,]$ means the commutator of operators, and $\pi$  is a representation to be described in Section \ref{dir}.

   Denote by \(\cont(\Omega)\) the algebra of continuous complex-valued functions  f of 
   \(\Omega\). Let \(p^{\prime} \geq
  1\) be the number implicitly defined given \(p \geq 1\) by \(\frac{1}{p} + \frac{1}{p^{\prime}} = 1\){\color{blue}\bf , }so that
  \(\lp{p^{\prime}}\) is the dual space of \(\lp{p}\). Typical continuous functions are denoted by \(f \in
  \cont(\Omega)\) and \(p\)-integrable functions by \(g \in \lp{p}(\bm{\mu})\). $\mathcal{K}$ denotes the Koopman
  operator and \(\ruelle\) denotes the Ruelle operator.

  The Koopman and Ruelle operators are characterized by:
  \begin{align} \label{yet}
    \koopman f \defn f \circ \sigma \text{, and: } \ruelle[f](x) \defn \frac{1}{2} \left( f(0x) + f(1x) \right)  {\color{blue}\bf \text{,}
  }\end{align}
  for all continuous functions \(f \in \cont(\Omega)\); they may be closed with respect to any \(p\)-norm, and
  we will use the same notation, \(\koopman\) and \(\ruelle\) still.

  General results on the Ruelle and Koopman operators can be found in \cite{PP}. They are dual of each other in the case of $L^2(\mu)$ (see \cite{BL})

    One of our goals (see end of Section \ref{dir}) is to show that
    \begin{equation} \norm{\left[ \mathcal{D}_p, \pi(\koopman \,\mathcal{L})]\right]}=1,
    \end{equation}
which is a particular case of
\begin{equation}\label{pen}\norm{\left[\mathcal{D} , \pi (\koopman^{n} \ruelle^{n})\right]} = 1.\end{equation}

   \section{Connes Distance}

  A Dirac operator $D$ is necessary to define a spectral triple and a Connes distance. In \cite{lpspec} the authors
  pose the following generalization for the definition of a spectral triple: 
  \begin{definition} \label{lp-spec-triples} 
    An \(\lp{p}\)-spectral triple is an ordered triple  \((A,\lp{p}(\mu),D)\), where:
    \begin{enumerate} \item \(\lp{p}(\mu)\) is an arbitrary \(\lp{p}\)-space;

    \item
       \(A\)  is an \(\lp{p}\)-operator algebra; \(\pi\) is a representation of   \(A\) on \(\lp{p}(\mu)\).

    \item
      \(D\) is an unbounded linear operator on \(\lp{p}(\mu)\), such that:
      \begin{enumerate}
          \item
               \(\{a \in A \mid \norm{\left[D,\pi(a)\right]} <+\infty \}\) is a norm dense subalgebra of  \(A\).

          \item \label{lpspec3a}
              \(\left(\id + D^{2}\right)^{-1}\) is a compact operator.

          \item \label{lpspec3b}
              For any complex \(\lambda\) not in the spectrum of \(D\), \(\left(D - \lambda \id\right)^{-1}\) is a compact operator.
      \end{enumerate}
    \end{enumerate}
  The operator \(D\) is called the Dirac operator.
  \end{definition}

  We are also going to follow \cite[Definition 3.3]{lpspec} and define the space of states of a given
  \(\lp{p}\)-operator algebra \(A\) as:
  \begin{align} \label{statesdef}
      \states(A) \defn \left\{\eta \in A^{\prime} \mid \norm{\eta} = \eta(1) = 1 \right\}  {\color{blue}\bf \text{.}
  }\end{align}

  The Connes distance between a pair of states \(\eta, \xi \in \states(A)\) is defined as:
  \begin{align} \label{connesdef}
    d_{D}(\eta, \xi) \defn \sup_{\substack{a \in A \\ \norm{\left[D,\pi(a)\right]} \leq 1}} \abs{\eta(a) - \xi(a)}  {\color{blue}\bf \text{.}
  }\end{align}
   It is an operator algebra version of the Wasserstein distance (see \eqref{wawa}).

  In  our setting, a continuous function   $f \in \cont(\Omega)$ is represented by the   bounded linear operator $\mult_f \in
  \bound(\lp{p}(\bm{\mu}))$ which acts on $L^p(\mu)$ by $\mult_f(g)=f g$. We will exploit this choice to introduce a form of
  Dirac operator $D$ is defined in terms of the shift dynamics over \(\Omega\) (see Section \ref{dir}).

  Notice the parameter $p$ governs on which space   $a \in A = \mathcal{C}(\Omega)$ is being represented by \(\pi\) and
  \(D\) is acting on. 

   For the more precise estimates of $d_D$ in Section \ref{purestates} and \ref{gen} note that the states in
  \eqref{connesdef} are Borel probabilities on $\Omega$. In this case, when computing \eqref{connesdef}, given $\eta$,
  we get that $\eta(a)=\eta(f)= \int f \, d\, \eta$, when   $a=f\in A=\mathcal{C}(\Omega).$ That is, $ \abs{\eta(a)
  - \xi(a)}= \abs{\eta(f) - \xi(f)}=| \int f d \eta- \int f d \xi|,$ and $d_D$ in \eqref{connesdef} is defined
  accordingly. Among other results we will estimate $d_p(\delta_{x}, \delta_{\sigma(x)})$ and $d_p(\delta_{x}, \delta_y)
  $, when $x,y\in \Omega$; and also $d_p(\eta,\xi).$ 

   \begin{remark} \label{ifi}
    Note also the importance in each case to estimate   whether or not $\norm{\left[D,\pi(a)\right}\leq 1$ for a certain
    given $a$. This helps to find lower bounds for $ d_{D}(\eta, \xi) $.
   \end{remark}

  \begin{remark}
    In what follows we are mostly interested in explicit computations and bounds for Connes pseudometric distance in the
    space of states of the \(\lp{p}\)-operator algebra \(\cont(\Omega)\). In \cite[Proposition 3, 4]{cfrconnes} A. Connes 
    notes that for the purposes of defining a pseudometric, \(D\) is not required to have compact resolvent. In fact, in
    \cite{state-metrics} M. Rieffel describes a considerably more general setting in which an analog of the Connes
    distance may be defined, and we are going to show that his setting includes ours in Section \ref{purestates}. In
    particular \cite[Proposition 1.4]{state-metrics} can be used to show that our pseudometric induces a \textit{strictly} 
    finer topology than the weak-\(*\) topology on \(\states(\cont(\Omega))\). Section \ref{purestates} will also
    provide insight into this matter as it describes the connected components of this topology. 
  \end{remark}

  \section{The Dirac Operator} \label{dir}
  
  Let \(A := \cont(\Omega)\). In this section we frequently identify a continuous function  $f \in \cont(\Omega)$ 
  with the bounded linear operator $\mult_f \in \bound(\lp{p}(\bm{\mu}))$. In this way, we often think of \(\pi\) as
  a representation of \(\cont(\Omega)\), while, rigorously, it is \(\pi \circ \mult_{(\,)}\) that is so.

  Let  \(\bound(\lp{p}(\bm{\mu}))\) act on \(\lp{p}(\bm{\mu} \sqcup \bm{\mu}) \cong \lp{p}(\bm{\mu}) \times \lp{p}(\bm{\mu})\)
  via a diagonal representation  \(\pi : \bound(\lp{p}(\bm{\mu})) \to \bound(\lp{p}(\bm{\mu}) \times \lp{p}(\bm{\mu}))\), 
  in such away that given  $f \in \mathcal{C}(\Omega)$:
 \begin{align*}
      \pi( \mult_f) \defn  \begin{bmatrix} \mult_f & 0 \\ 0 & \mult_f \end{bmatrix}  \text{,}
  \end{align*}
  and let \(D = \dirac_{p}\) be the linear operator acting on \(\lp{p}(\bm{\mu}) \times \lp{p}(\bm{\mu})\) by:
   \begin{equation}  \label{jiu}
      \dirac_{p} \defn \begin{bmatrix} 0 & \koopman \\ \ruelle & 0 \end{bmatrix} \text{.}
  \end{equation}

   In \cite{CHLS} the authors considered other forms of dynamically defined Dirac operators.

  Here  the states are defined by \begin{equation*}  \label{statesdef27}
      \states(A) \defn \left\{\eta \in A^{\prime} \mid \norm{\eta} = \eta(1) = 1 \right\} \text{,}
  \end{equation*}
   and the Connes distance for   $\eta,\xi \in \states(A)$ by:
  \begin{equation}  \label{connesdefu}
  d_p (\eta, \xi) = d_{\mathcal{D}_p}(\eta, \xi) \defn \sup_{\substack{a \in A \\ \norm{\left[\mathcal{D}_p,\pi(a)\right]} \leq 1}} \abs{\eta(a) - \xi(a)} \text{.}
  \end{equation}

   In order to compute expression \eqref{connesdefu} it helps to know which \(f \in \cont(\Omega)\) satisfies
  \(\norm{\left[\dirac_{p}, \pi(\mult_f)\right]} \leq 1\). We will present here some explicit estimates that will allow us
  to derive lower bounds for the Connes distance when $D =\mathcal{D}_p $. As a first step in that direction, notice
  that for any \(f \in \cont(\Omega)\):
  \begin{align} \label{wtransf}
      \left[ \dirac_{p} , \pi( \mult_f) \right] & =
      \begin{bmatrix}
          0 & \koopman \\
          \ruelle & 0
      \end{bmatrix}  \begin{bmatrix}
                \mult_f & 0 \\
                        0 & \mult_f
      \end{bmatrix}
      -  \begin{bmatrix}
                \mult_f & 0 \\
                        0 & \mult_f
      \end{bmatrix} \begin{bmatrix}
          0 & \koopman \\
          \ruelle & 0
      \end{bmatrix} \nonumber \\
                                     & =  \begin{bmatrix}
                                         0 & \koopman \mult_f - \mult_f \koopman \\
\ruelle \mult_f - \mult_f \ruelle & 0
                                     \end{bmatrix} \nonumber \\
                                     & =  \begin{bmatrix}
                                         0 & \mult_{f\circ \sigma - f} \koopman \\
\ruelle \mult_{f- f \circ \sigma} & 0
                                     \end{bmatrix}  {\color{blue}\bf \text{.}
  }\end{align}
  Consequently:
  \begin{align} \label{pcommnorm}
       \norm{\left[ \dirac_{p} , \pi(\mult_f) \right]} & = \max \left\{
               \norm{\mult_{f\circ \sigma - f} \koopman}_{p},
               \norm{\ruelle \mult_{f- f \circ \sigma}}_{p} \right\} \nonumber \\
                                                & = \max \left\{
               \norm{\mult_{f\circ \sigma - f} \koopman}_{p},
               \norm{\mult_{f\circ \sigma - f} \koopman}_{p^{\prime}} \right\} \nonumber \\
                                                & = \max_{\lambda \in \left\{p, p^{\prime}\right\}}  \norm{\mult_{f\circ \sigma - f} \koopman}_{\lambda}  {\color{blue}\bf \text{.}
  }\end{align}

  Equation \eqref{wtransf} shows that the derivative of a given function \(f \in \cont(\Omega)\) with respect to
  \(\dirac_{p}\) is completely characterized by a weighted transfer  operator with weight given by a discrete-time
  forward dynamical derivative of \(f\), namely   \(\mult_{f\circ \sigma - f} \koopman\). Then, the theory of weighted
  transfer operators applies (see \cite{tentropy} and \cite{lpcuntz}). In particular, there is a lower bound for the
  ``Lipschitz Seminorm''   \(\norm{\left[ \dirac_{p} , \pi(\mult_f) \right]}\) given by the variational principle for the
  spectral radius:

  \begin{proposition} \label{lpcons}
    For any continuous function $f:\Omega \to \mathbb{R}$  
     \begin{equation*}
      \norm{\left[ \dirac_{p} , \pi(\mult_f) \right]} \geq r(\mult_{f\circ \sigma - f} \koopman) =
    \end{equation*}
    \begin{equation} \label{puq1} 
      \max_{\lambda \in \left\{p, p^{\prime}\right\}}\sqrt[\lambda]{2} \sup_{\mu \in \prob(\sigma)} \exp \left( \int \log \abs{f \circ \sigma - f} \d \mu
                                            + \frac{h_{\mu}(\sigma)}{\lambda} \right)  {\color{blue}\bf \text{.}
    }\end{equation}
  \end{proposition}
  \begin{proof}
    Apply \cite{lpcuntz} or \cite{specweightiso}. 
  \end{proof}

  \begin{remark}
    When \(p = 2\), the equality holds, since the norm of an anti-selfadjoint operator is equal to its spectral
    radius. Note also that 
     \begin{equation*} 
      \frac{1}{p} \sup_{\mu \in \prob(\sigma)}\int p\, \log \abs{f \circ \sigma - f} d \mu + h(\mu)
    \end{equation*} 
    is not exactly the classical Pressure problem (as in \cite{PP}) due to the fact that $\log \abs{f \circ \sigma - f}$ 
    can take the value $- \infty$.
  \end{remark}

  \begin{remark}
    Combining Proposition \ref{lpcons} and Birkhoff's ergodic theorem, it follows that for \(\bm{\mu}\)-almost every 
    \(x \in \Omega\):
    \begin{align*}
       \norm{\left[ \dirac_{p} , \pi(\mult_f) \right]} & \geq \exp \int \log \abs{f \circ \sigma - f} \d \bm{\mu} \\
        & = \exp \sum_{n = 0}^{+\infty} \log \abs{f \circ \sigma^{n+1}(x) - f \circ
        \sigma^{n}(x)} \\
        & = \prod_{n = 0}^{+\infty} \abs{f \circ \sigma^{n+1}(x) - f \circ
        \sigma^{n}(x)}  {\color{blue}\bf \text{.}
    }\end{align*}
  \end{remark}

  Some of the present results can also be deduced from the abstract point of view of \cite{lpcuntz}, \cite{specweightiso} 
  or \cite{tentropy}. For example, the reader should compare \eqref{pdrms1} and \cite[Equation (98)]{tentropy}

  \begin{lemma}
    For any $f \in \cont(\Omega)$:
    \begin{equation*}
      \left[ \dirac_{p} , \pi( \mult_f) \right] = 0 \iff f \circ \sigma - f = 0   \text{.}
    \end{equation*}
    The latter implies that $f$ is constant.
  \end{lemma}

  \begin{proof}
    The proof is analogous to the one in \cite{BL}. If \(f \circ \sigma - f = 0\), then:
    \begin{equation*}
         \norm{\left[ \dirac_{p} , \pi(\mult_f) \right]} = \max_{\lambda \in \left\{p, p^{\prime}\right\}}  \norm{\mult_{f\circ \sigma - f} \koopman}_{\lambda} = \max_{\lambda \in \left\{p, p^{\prime}\right\}} \norm{\mult_{0} \koopman}_{\lambda} = 0  {\color{blue}\bf \text{.}
    }\end{equation*}
    In the other direction, if  \(\left[ \dirac_{p} , \pi(\mult_f) \right] = 0\), then:
    \begin{align*}
      \max_{\lambda \in \left\{p, p^{\prime}\right\}}  \norm{\mult_{f\circ \sigma - f} \koopman}_{\lambda} = 0  {\color{blue}\bf \text{,}
    }\end{align*}
    and in particular:
    \begin{align*}
      \max_{\lambda \in \left\{p, p^{\prime}\right\}}  \abs{\mult_{f\circ \sigma - f} \koopman (1)}_{\lambda} & = \max_{\lambda \in \left\{p, p^{\prime}\right\}} \abs{f \circ \sigma - f}_{\lambda} \\
                                                                                                         & = 0  {\color{blue}\bf \text{,}
    }\end{align*}
    which means \(f \circ \sigma - f = 0\).
      \end{proof}

  \begin{lemma} \label{pinftyseminorm}
    For any \(f \in \cont(\Omega)\):
    \begin{align*}
      \abs{f}_{\infty} \geq \sup_{\abs{g}_{p} = 1} \abs{f \koopman g}_{p} \geq \abs{\ruelle f}_{\infty}  {\color{blue}\bf \text{.}
    }\end{align*}
    Furthermore, if \(f \in \cont(\Omega)\) does not depend on the first coordinate (that is, if \(f\) is
    \(\sigma^{-1}(\Sigma)\)-measurable), then all above inequalities are equalities.
  \end{lemma}
  \begin{proof}
    The proof is similar to the one in \cite{BL}, except for the convex function to which we apply Jensen's
    inequality for conditional expectations is now \(\abs{\cdot}^{p}\) instead of \(\abs{\cdot}^{2}\).
  \end{proof}
  \begin{remark} \label{pandp}
    Lemma \ref{pinftyseminorm} holds for \(p\) and \(p^{\prime}\) with the same bounds.
  \end{remark}

  \begin{theorem} \label{ptrocd}
    Replacing \(f\) by \(f \circ \sigma - f\) in Lemma \ref{pinftyseminorm}, in view of \eqref{pcommnorm} and Remark
    \ref{pandp}, we get for any \(f \in \cont(\Omega)\):
  \begin{align*}
    \abs{\koopman f - f}_{\infty} & \geq  \norm{\left[ \dirac_{p} , \pi(\mult_f) \right]} \geq \abs{f - \ruelle f}_{\infty}  {\color{blue}\bf \text{.} }\\
    \intertext{Moreover, if \(f \circ \sigma - f\) does not depend on the first coordinate we get the equalities:}
    \abs{\koopman f - f}_{\infty} & =  \norm{\left[ \dirac_{p} , \pi(\mult_f) \right]} = \abs{f - \ruelle f}_{\infty}  {\color{blue}\bf \text{.}
  }\end{align*}
  \end{theorem}

  \begin{proposition}
    For any \(f \in \cont(\Omega)\):
    \begin{equation} \label{pdrms1}
       \norm{\left[ \dirac_{p} , \pi(\mult_f) \right]} = \max_{\lambda \in \left\{p, p^{\prime}\right\}} \abs{\sqrt[\lambda]{\ruelle \abs{f \circ \sigma - f}^{\lambda}}}_{\infty}  {\color{blue}\bf \text{.}
    }\end{equation}

    Expression \eqref{pdrms1} can be written as:
    \begin{equation} \label{pdrms2}
      \max_{\lambda \in \left\{p, p^{\prime}\right\}} \abs{\sqrt[\lambda]{\ruelle \abs{f \circ \sigma - f}^{\lambda}}}_{\infty} = \max_{\lambda \in \left\{p, p^{\prime}\right\}} \sup_{x \in \Omega} \sqrt[\lambda]{\begin{array}{l}
  \frac{\abs{f(x) - f(0x)}^{\lambda}}{2} + \\ \quad + \frac{\abs{f(x) - f(1x)}^{\lambda}}{2}
      \end{array}}  {\color{blue}\bf \text{.}
    }\end{equation}

    The right-hand side of \eqref{pdrms2} is a form of the supremum of {\it mean backward derivative}.
  \end{proposition}
  \begin{proof}
    Analogous to \cite{rkboson}. We have:
  \begin{align*}
    \sup_{\abs{g}_{\lambda} = 1} \abs{f \koopman g}_{\lambda} & = \sup_{\abs{g}_{\lambda} = 1} \left( \int \abs{f \koopman g}^{\lambda} \d \bm{\mu} \right)^{\frac{1}{\lambda}} \\
                                             & = \sup_{\abs{g}_{\lambda} = 1} \left( \int \abs{f}^{\lambda} \abs{\koopman g}^{\lambda} \d \bm{\mu} \right)^{\frac{1}{\lambda}} \\
                                             & = \sup_{\abs{g}_{\lambda} = 1} \left( \int \abs{f}^{\lambda} \left( \koopman \abs{g}^{\lambda} \right) \d \bm{\mu} \right)^{\frac{1}{\lambda}} \\
                                             & = \sup_{\abs{g}_{\lambda} = 1} \left( \int \left( \ruelle \abs{f}^{\lambda} \right) \abs{g}^{\lambda} \d \bm{\mu} \right)^{\frac{1}{\lambda}} \\
                                             & = \sup_{\abs{g}_{\lambda} = 1} \abs{\left( \sqrt[\lambda]{\ruelle \abs{f}^{\lambda}} \right) g}_{\lambda} \\
                                             & = \abs{\sqrt[\lambda]{\ruelle \abs{f}^{\lambda}}}_{\infty}  {\color{blue}\bf \text{,}
  }\end{align*}
  then we substitute \(f\) for \(f \circ \sigma - f\).
  \end{proof}

  \begin{corollary} \label{kolmomeans}
    Notice that:
    \begin{align*}
\max_{\lambda \in \left\{p, p^{\prime}\right\}} \sqrt{\begin{array}{l}
\frac{\abs{f(x) - f(0x)}^{\lambda}}{2} + \\ \quad + \frac{\abs{f(x) - f(1x)}^{\lambda}}{2}
    \end{array}} =
    \sqrt[\max\left\{p, p^{\prime}\right\}]{\begin{array}{l}
\frac{\abs{f(x) - f(0x)}^{\max\left\{p, p^{\prime}\right\}}}{2} + \\ \quad + \frac{\abs{f(x) - f(1x)}^{\max\left\{p, p^{\prime}\right\}}}{2}
    \end{array}}  {\color{blue}\bf \text{,}
    }\end{align*}
    and that \(\max\left\{p, p^{\prime}\right\} \geq 2\), so:
    \begin{align*}
     \min \left\{\begin{array}{l} \abs{f(x) - f(0x)}, \\ \quad \abs{f(x) - f(1x)} \end{array}\right\} & \leq \frac{2}{\frac{1}{\abs{f(x) - f(0x)}} + \frac{1}{\abs{f(x) - f(1x)}}} \\
                                                                                                                    & \leq \sqrt{\begin{array}{l}
\abs{f(x) - f(0x)} \times \\ \quad \times \abs{f(x) - f(1x)}
                                                                                                                    \end{array} } \\
                                                                                    & \leq \frac{1}{2} \left(\begin{array}{l}
\abs{f(x) - f(0x)} + \\ \quad + \abs{f(x) - f(1x)}
                                                                                    \end{array} \right) \\
                                                                                    & \leq \sqrt{\begin{array}{l}
\frac{\abs{f(x) - f(0x)}^{2}}{2} + \\ \quad + \frac{\abs{f(x) - f(1x)}^{2}}{2}
    \end{array}} \\
                                                                                    & \leq \sup_{x \in \Omega} \sqrt{\begin{array}{l}
\frac{\abs{f(x) - f(0x)}^{2}}{2} + \\ \quad + \frac{\abs{f(x) - f(1x)}^{2}}{2}
    \end{array}} \\
                                                                                    & \leq  \norm{\left[ \dirac , \pi(\mult_f) \right]} \\
                                                                                    & \leq \sup_{x \in \Omega} \max \left\{\begin{array}{l}
\abs{f(x) - f(0x)}, \\ \quad \abs{f(x) - f(1x)}
                                                                                    \end{array} \right\} \\
                                                                                    & = \abs{f \circ \sigma - f}_{\infty}  {\color{blue}\bf \text{.}
    }\end{align*}
    This reasoning also provides an alternative proof for the first inequality in Theorem \ref{ptrocd}.
  \end{corollary}
  \begin{remark}
    Corollary \ref{kolmomeans} shows that:
    \begin{align*}
         \norm{\left[ \dirac_{p} , \pi(\mult_f) \right]} & = \max_{\lambda \in \left\{p, p^{\prime}\right\}}  \norm{\mult_{f\circ \sigma - f} \koopman}_{\lambda} \\
        & =  \norm{\mult_{f\circ \sigma - f} \koopman}_{\max \left\{p, p^{\prime}\right\}}  \bf \text{.}
    \end{align*}
    Henceforth, we set \(\lambda \defn \max \left\{p, p^{\prime}\right\}\), as this will cause no confusion.
  \end{remark}

   To conclude the characterization of the functions \(f \in \cont(\Omega)\) that have 
    \(\norm{\left[ \dirac_{p} , \pi(\mult_f) \right]} \leq 1\), first we will exhibit a sufficient condition:
  \begin{proposition}
    For any function \(f \in \cont(\Omega)\):
    \begin{align*}
      \abs{f \circ \sigma - f}_{\infty} \leq 1 \implies  \norm{\left[ \dirac_{p}, \pi(\mult_f) \right]} \leq 1
    \end{align*}
  \end{proposition}
  \begin{proof}
    Apply Theorem \ref{ptrocd}.
  \end{proof}

  And, lastly, we present a necessary condition:
  \begin{proposition} \label{pointest}
    For any function \(f \in \cont(\Omega)\):
  \begin{align*}
     \norm{\left[ \dirac_{p}, \pi(\mult_f) \right]} \leq 1 \implies \abs{f\circ\sigma - f}_{\infty} & \leq \sqrt[\lambda]{2}  {\color{blue}\bf \text{.}
  }\end{align*} 
  \end{proposition}
  \begin{proof}
    Notice that for any \(x \in \Omega\):
  \begin{align*}
      & &  \norm{\left[ \dirac_{p} , \pi(\mult_f) \right]} & \leq 1 \\
      \iff & & \sqrt[\lambda]{\begin{array}{l}
\frac{\abs{f(x) - f(0x)}^{\lambda}}{2} \\ \quad + \frac{\abs{f(x) - f(1x)}^{\lambda}}{2}
      \end{array}} & \leq 1 \\
\iff & & \frac{\abs{f(x) - f(0x)}^{\lambda}}{2} & \leq 1 - \frac{\abs{f(x) - f(1x)}^{\lambda}}{2} \\
\iff & & \abs{f(x) - f(0x)}^{\lambda} & \leq 2 \left(1 - \frac{\abs{f(x) - f(1x)}^{\lambda}}{2}\right) \\
\iff & & \abs{f(x) - f(0x)}^{\lambda} & \leq 2 - \abs{f(x) - f(1x)}^{\lambda} \\
\implies & & \abs{f(x) - f(0x)} & \leq \sqrt[\lambda]{2}  {\color{blue}\bf \text{.}
  }\end{align*}
  Then, exchanging \(0\) and \(1\) in the previous argument we prove our main claim.
  \end{proof}

   The specific form of \(\dirac_{p}\) we are considering also makes it convenient to compute the Lipschitz seminorm for
  operators of the form \(\koopman^{n} \ruelle^{n}\).

    We will show that
    \begin{equation} \norm{\left[ \mathcal{D}_p, \pi(\koopman \,\mathcal{L})]\right]}=1,
    \end{equation}
 which is a particular case of   
\begin{equation}\label{pen}\norm{\left[\mathcal{D} , \pi (\koopman^{n} \ruelle^{n})\right]} = 1.\end{equation}

In order to get that  all the elements in the above expressions are well defined, we consider the identification of $f$ with $M_f$.

    In order to show \eqref{pen}, first notice the same computations at the end of \cite[Section 2]{rkboson} hold for
    \(\mathcal{D}_{p}\). To recall:
    \begin{align*}
      \left[\mathcal{D}_{p} , \pi (\koopman^{n} \ruelle^{n})\right] & = \begin{pmatrix}
            0 & \koopman \koopman^{n} \ruelle^{n} - \koopman^{n} \ruelle^{n} \koopman \\
            \ruelle \koopman^{n} \ruelle^{n} - \koopman^{n} \ruelle^{n} \ruelle & 0 \\
        \end{pmatrix} \\
                                                                   & = \begin{pmatrix}
            0 & \koopman \koopman^{n} \ruelle^{n} - \koopman^{n} \ruelle^{n-1} \\
            \koopman^{n-1} \ruelle^{n} - \koopman^{n} \ruelle^{n} \ruelle & 0 \\
        \end{pmatrix} \\
                                                                   & = \begin{pmatrix}
                                                                       0 & \koopman \left( \koopman^{n} \ruelle^{n} - \koopman^{n-1} \ruelle^{n-1} \right) \\
                                                                       \left( \koopman^{n-1} \ruelle^{n-1} - \koopman^{n} \ruelle^{n} \right) \ruelle & 0 \\
        \end{pmatrix} \text.
    \end{align*}

    Furthermore:
    \begin{align*}
        \left( \koopman^{n-1} \ruelle^{n-1} - \koopman^{n} \ruelle^{n} \right)^{2} & = \begin{array}{l}
            \koopman^{n-1} \ruelle^{n-1} \koopman^{n-1} \ruelle^{n-1} - \koopman^{n-1} \ruelle^{n-1} \koopman^{n} \ruelle^{n} \\
            \quad - \koopman^{n} \ruelle^{n} \koopman^{n-1} \ruelle^{n-1} + \koopman^{n} \ruelle^{n} \koopman^{n} \ruelle^{n}
        \end{array} \\
                                                                                                 & = \begin{array}{l}
            \koopman^{n-1} \ruelle^{n-1} - \koopman^{n} \ruelle^{n} \\
            \quad - \koopman^{n} \ruelle^{n} + \koopman^{n} \ruelle^{n}
        \end{array} \\
                                                                                                 & = \koopman^{n-1} \ruelle^{n-1} - \koopman^{n} \ruelle^{n} \text.
    \end{align*}

    Also, notice that, because \(\bm{\mu}\) is invariant, \(\koopman : \lp{p} \to \lp{p}\) is an isometry for any 
    \(1 \leq p \leq +\infty\). Therefore, \(\norm{\koopman T} = \norm{T}\) for any bounded linear transformation \(T
    : \lp{p} \to \lp{p}\). In particular, when \(T\) is a projection, such as \( \koopman^{n-1} \ruelle^{n-1} - \koopman^{n} \ruelle^{n} \),
  \(\norm{\koopman T} = 1\). This shows \(\norm{\left[\mathcal{D} , \pi (\koopman^{n} \ruelle^{n})\right]} = 1\).

   We consider this fact important because it could be a starting point to eventually computing the Connes distance
  induced by the \(\dirac_{p}\) on the space of states of a C*-algebra such as the Exel-Lopes C*-algebra (see \cite{EL2}), since it is
  generated by elements of the form:
  \begin{equation*}
     \sum_{i = 1}^{N} \mult_{f_{n_{i}}} \koopman^{n_{i}} \ruelle^{n_{i}} \mult_{g_{n_{i}}} \text{,}
  \end{equation*}
  where \(N, n_{i} \in \natural\), and \(f_{n_{i}}, g_{n_{i}} \in \cont(\Omega)\), for every \(1 \leq i \leq N\).
   \bigskip
  
  \section{Pure States\texorpdfstring{ - $d_p(\delta_{x}, \delta_{y})$}{}} \label{purestates}

  In the particular case of the \(\lp{p}\)-operator algebra \(\cont(\Omega)\) the set of states (defined in
  \eqref{statesdef}) is exactly the same as for the C*-algebra \(\cont(\Omega)\). That is because the
  \(\lp{p}\)-operator norm of a multiplication operator   \(\mult_f\) is \(\abs{f}_{\infty}\) regardless of which \(p\) one
  chooses. In this case  $A= \mathcal{C}(\Omega);$ and the states are the Borel probability measures defined
  on \(\Omega\). In this section, we consider the particular case of pure states: the Dirac deltas \(\delta_{x}\) on
  points \(x \in \Omega\). In the next section, we will consider a more general case.

  A natural question is to know when $d_p(\delta_{x}, \delta_y)<\infty$, for $x,y\in \Omega$ (see Theorem
  \ref{acima}); which is somehow related to homoclinic equivalence relations.

  There is a way to fit our results into the setting of \cite{state-metrics}: in their notation, our normed space is
  \(A = \cont(\Omega)\) equipped with the supremum norm \(\abs{\cdot}_{\infty}\). All of our elements are Lipschitz,
  so \(\mathcal{L} = A = \cont(\Omega)\). Our Lipschitz seminorm is given by \(L(a) = \norm{\left[\dirac_{p},
  \pi(a)\right]}\). Its zero locus is the space of constant functions \(\mathcal{K} = \complex 1 \subseteq
  \cont(\Omega)\), which determines \(\eta\) up to sign. Take it so \(\eta(1) = 1\).

  This means that we could apply \cite[Proposition 1.4]{state-metrics}, which says our distance induces a topology
  finer than weak-\(*\). Additionally, we will see in Example \ref{infty} that we do have states for which the
  distance is \(+\infty\), which means it induces a \textit{strictly} finer topology than the weak-\(*\) topology in
  \(\states(\cont(\Omega))\). Notice this topology has many nontrivial connected components.


  We now pass to the study of the connected components of the Connes distance. This means we want to discriminate
  between the pairs of states for which it is finite and the pairs of states for which it is not. First, we will
  restrict our attention to pure states. We begin with a simple example:

  \begin{example} \label{ddxdsx}
    Let us estimate the Connes distance $d_p$ for a pair of Dirac deltas \(\delta_{x}, \delta_{\sigma(x)} \in \states(\cont(\Omega))\).
    Proposition \ref{pointest} implies: when $\lambda = \max\{p, p^\prime\}$
    \begin{align*}
     d_p(\delta_{x}, \delta_{\sigma(x)}) & = \sup_{\substack{f \in \cont(\Omega) \\ \norm{\left[\dirac_{p}, \pi(\mult_f)\right]} \leq 1}} \abs{f(x) - f \circ \sigma(x)} \\
                                          & \leq \sqrt[\lambda]{2} \\
                                          & < +\infty  {\color{blue}\bf \text{.}
    }\end{align*}
  \end{example}

  \begin{example} \label{smxsny}
    Consequently, if \(x \in \Omega\) and \(y \in \Omega\) are two points such that there exists two numbers 
    \(m, n \in \natural\) for which the respective orbits meet at \linebreak \(\sigma^{m}(x) = \sigma^{n}(y)\), then:
  \begin{align*}
    d_p(\delta_{x}, \delta_{y}) & \leq d_p(\delta_{x}, \delta_{\sigma(x)}) + d_p(\delta_{\sigma(x)}, \delta_{y}) \\
                              & \leq d_p(\delta_{x}, \delta_{\sigma(x)}) + d_p(\delta_{\sigma(x)}, \delta_{\sigma^{2}(x)}) + d_p(\delta_{\sigma^{2}(x)}, \delta_{y}) \\
                              & \leq d_p(\delta_{x}, \delta_{\sigma(x)}) + \cdots + d_p(\delta_{\sigma^{m-1}(x)}, \delta_{\sigma^{m}(x)}) + d_p(\delta_{\sigma^{m}(x)}, \delta_{y}) \\
                              & = d_p(\delta_{x}, \delta_{\sigma(x)}) + \cdots + d_p(\delta_{\sigma^{m-1}(x)}, \delta_{\sigma^{m}(x)}) + d_p(\delta_{\sigma^{n}(y)}, \delta_{y}) \\
                              & \leq  \begin{array}{l}
p(\delta_{x}, \delta_{\sigma(x)}) + \cdots + d_p(\delta_{\sigma^{m-1}(x)}, \delta_{\sigma^{m}(x)}) + \\
uad + d_p(\delta_{\sigma^{n}(y)}, \delta_{\sigma^{n-1}(y)}) + \cdots + d_p(\delta_{\sigma^{1}(y)}, \delta_{y})
                              \end{array} \\
                              & \leq \sqrt[\lambda]{2} \left( m + n \right) \\
                              & < +\infty  {\color{blue}\bf \text{.}
  }\end{align*}
  \end{example}

  Now let us calculate more examples. We are looking for a function of arbitrary variation and ``Lipschitz constant \(= 1\)''.
  \begin{example}
    If \( f \defn 2 \chi_{01} + 4 \chi_{11} + 2 \chi_{10}\), then:
  \begin{align*}
    f \circ \sigma - f & =  \begin{array}{l}
\left( \chi_{001} + \chi_{101} + \chi_{010} + \chi_{110} \right) + 4 \left( \chi_{011} + \chi_{111} \right) + \\
uad - 2 \left(\chi_{010} + \chi_{011} + \chi_{100} + \chi_{101}\right) - 4 \left(\chi_{110} + \chi_{111}\right)
    \end{array} \\
                       & = 2 \left( \chi_{001} - \chi_{100} - \chi_{011} + \chi_{110} \right) + 4 \left(\chi_{011} - \chi_{110}\right) \\
                       & = 2 \left( \chi_{001} - \chi_{100} + \chi_{011} - \chi_{110} \right)  {\color{blue}\bf \text{.}
  }\end{align*}
  \end{example}

  \begin{example}
    If \( f \defn 2 \chi_{001} + 4 \chi_{011} + 6 \chi_{111} + 4 \chi_{110} + 2 \chi_{100} \), then:
  \begin{align*}
    f \circ \sigma - f & = \begin{array}{l}
        2 \left( \chi_{0001} + \chi_{1001} + \chi_{0100} + \chi_{1100} \right) \\
        \quad + 4 \left( \chi_{0011} + \chi_{1011} + \chi_{0110} + \chi_{1110}\right) \\
        \quad \quad + 6 \left(\chi_{0111} + \chi_{1111}\right) \\
        \quad - 2 \left(\chi_{0010} + \chi_{0011} + \chi_{1000} + \chi_{1001}\right) \\
        \quad \quad - 4 \left(\chi_{0110} + \chi_{0111} + \chi_{1100} + \chi_{1101}\right) \\
        \quad \quad \quad - 6\left(\chi_{1110} + \chi_{1111}\right)
    \end{array} \\
                       & = \begin{array}{l}
        2 \left( \chi_{0001} + \chi_{0100} + \chi_{1100} - \chi_{0010} - \chi_{0011} - \chi_{1000}\right) \\
        \quad + 4 \left( \chi_{0011} + \chi_{1011} + \chi_{1110} - \chi_{0111} - \chi_{1100} - \chi_{1101}\right) \\
        \quad \quad + 6 \left(\chi_{0111} - \chi_{1110}\right)
    \end{array} \\
                       & = \begin{array}{l}
                           2 \left( \chi_{0001} + \chi_{0100} - \chi_{1100} + \chi_{0111} - \chi_{0010} + \chi_{0011} - \chi_{1000} - \chi_{1110}\right) \\
        \quad + 4 \left(\chi_{1011} - \chi_{1101}\right) \\
    \end{array}  {\color{blue}\bf \text{.}
  }\end{align*}
  \end{example}

  \begin{example}
    If \(f \defn 2 \left(\chi_{001} + \chi_{010} + \chi_{100}\right) + 4 \left(\chi_{011} + \chi_{101} + \chi_{110}\right) + 6 \chi_{111}\), then:
  \begin{align*}
    f \circ \sigma - f & = \begin{array}{l}
        2 \left( \chi_{0001} + \chi_{1001} + \chi_{0010} + \chi_{1010} + \chi_{0100} + \chi_{1100} \right) \\
        \quad + 4 \left( \chi_{0011} + \chi_{1011} + \chi_{0101} + \chi_{1101} + \chi_{0110} + \chi_{1110}\right) \\
        \quad \quad + 6 \left(\chi_{0111} + \chi_{1111}\right) \\
        \quad - 2 \left(\chi_{0010} + \chi_{0011} + \chi_{0100} + \chi_{0101} + \chi_{1000} + \chi_{1001}\right) \\
        \quad \quad - 4 \left(\chi_{0110} + \chi_{0111} + \chi_{1010} + \chi_{1011} + \chi_{1100} + \chi_{1101}\right) \\
        \quad \quad \quad - 6\left(\chi_{1110} + \chi_{1111}\right)
    \end{array} \\
                       & = \begin{array}{l}
        2 \left( \chi_{0001} + \chi_{1010} + \chi_{1100} - \chi_{0101} - \chi_{0011} - \chi_{1000}\right) \\
        \quad + 4 \left( \chi_{0011} + \chi_{0101} + \chi_{0110} + \chi_{1110} - \chi_{0110} - \chi_{0111} - \chi_{1010} - \chi_{1100}\right) \\
        \quad \quad + 6 \left(\chi_{0111} - \chi_{1110}\right)
    \end{array} \\
                       & = \begin{array}{l}
                           2 \left( \chi_{0001} + \chi_{1010} + \chi_{1100} - \chi_{0101} - \chi_{0011} - \chi_{1000} + \chi_{0111} - \chi_{1110}\right) \\
        \quad + 4 \left( \chi_{0011} + \chi_{0101} + \chi_{0110} - \chi_{0110} - \chi_{1010} - \chi_{1100}\right) \\
    \end{array} \\
                       & = 2 \left( \chi_{0001} - \chi_{1010} + \chi_{1100} + \chi_{0101} - \chi_{0011} - \chi_{1000} + \chi_{0111} - \chi_{1110}\right)  {\color{blue}\bf \text{.}
  }\end{align*}
  \end{example}

  \begin{center}
  \includegraphics[width=\textwidth]{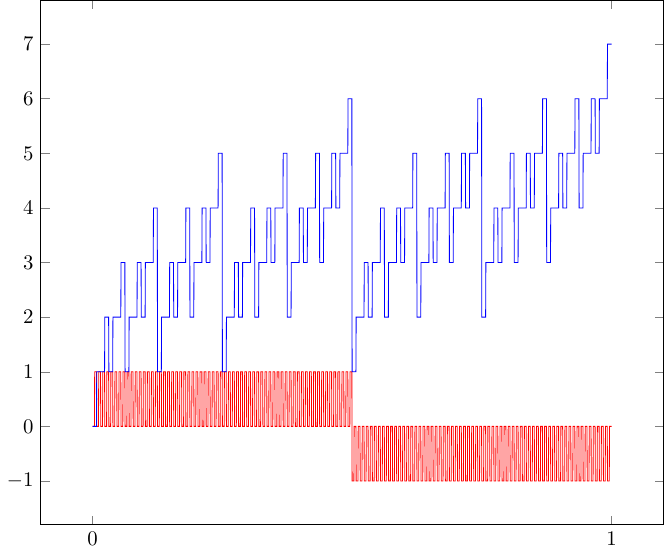} \label{fig:f-7}
  \end{center}
  \(f_{7}^{\gamma}\) (blue) and \(f_{7}^{\gamma} \circ \sigma - f_{7}^{\gamma}\) (red). In this picture, the sequence \(x \in \Omega\) 
  corresponds to the real number \(\sum x_{i} 2^{-i} \in [0,1]\).

  \begin{proposition}
    If \(x, y \in \Omega\) are two sequences such that:
    \begin{align*}
        \sup_{n \in \natural} \abs{\# \left\{i \leq n \mid x_{i} = 1\right\} - \# \left\{i \leq n \mid y_{i} = 1\right\}} \geq N  {\color{blue}\bf \text{,}
    }\end{align*}
    then \(d_p(\delta_{x}, \delta_{y}) \geq N\). In particular, \(d_p(\delta_{0^{\infty}}, \delta_{1^{\infty}}) = + \infty\).
  \end{proposition}
  \begin{proof}
    Consider the following family of continuous functions of \(\Omega\):
    \begin{align*}
        f_{k}^{\gamma} & \defn \sum_{w \in \hat{W}_{k}} \# \left\{i \mid w_{i} = 1\right\} \chi_{w}  {\color{blue}\bf \text{.}
    }\end{align*}
    Notice that:
    \begin{align*}
        \norm{\left[ \dirac_{p} , \pi(f_{k}^{\gamma}) \right]} & \leq \abs{a\left(f_{k}^{\gamma} \circ \sigma - f_{k}^{\gamma}\right)}_{\infty} \\
                                                  & \leq \abs{f_{k}^{\gamma} \circ \sigma - f_{k}^{\gamma}}_{\infty} \\
                                                  & \leq 1  {\color{blue}\bf \text{,}
    }\end{align*}
    so that this family gives us a lower bound for the Connes distance. That is:
    \begin{align*}
        d_p(\delta_{x}, \delta_{y}) & = \sup_{\substack{f \in \cont(\Omega) \\ \norm{\left[\dirac_{p}, \pi(\mult_f)\right]} \leq 1}} \abs{f(x) - f(y)} \\
                                  & \geq \sup_{k} \abs{f_{k}^{\gamma}(x) - f_{k}^{\gamma}(y)} \\
                                  & \geq N  {\color{blue}\bf \text{.}
    }\end{align*}
  \end{proof}
  \begin{example} \label{infty}
    In particular, \(d_p(\delta_{0^{\infty}}, \delta_{1^{\infty}}) = +\infty\).
  \end{example}

  \begin{proposition}
    If \(f_{k}\) is a function of the form \(\sum_{w \in \hat{W}_{k}} \theta_{w} \chi_{w}\) such that 
     \(\norm{\left[\dirac_{p}, \pi(\mult_f)\right]} \leq 1\), then \(\abs{f_{k}(x) - f_{k}(y)} \leq \sqrt[\lambda]{2} k\).
  \end{proposition}
  \begin{proof}
    Consider the point \(z = (x_{1}, x_{2}, \cdots, x_{k}, y_{1}, y_{2}, \cdots, y_{k}) \in \Omega\). It is clear that
    \(f_{k}(z) = f_{k}(x)\) and \(f_{k} \circ \sigma^{k}(z) = f_{k}(y)\). Now the values \(f_{k}(x)\) and \( f_{k}(y)\)
    are telescopically related as:
    \begin{align*}
        \abs{f_{k}(x) - f_{k}(y)} & = \abs{f_{k}(z) - f_{k} \circ \sigma^{k}(z)} \\
                                  & = \abs{f_{k}(z) - f_{k} \circ \sigma(z) + f_{k} \circ \sigma(z) - f_{k} \circ \sigma^{k}(z)} \\
                                  & = \abs{\begin{array}{l}
f_{k}(z) - f_{k} \circ \sigma(z) + \\
\quad + f_{k} \circ \sigma(z) - f_{k} \circ \sigma^{2}(z) + \\
\quad \quad + \cdots + \\
\quad \quad \quad + f_{k} \circ \sigma^{k-1}(z) - f_{k} \circ \sigma^{n}(z)
\end{array}} \\
                                  & \leq \begin{array}{l}
\abs{f_{k}(z) - f_{k} \circ \sigma(z)} + \\
\quad + \abs{f_{k}(\sigma(z)) - f_{k} \circ \sigma(\sigma(z))} + \\
\quad \quad + \cdots + \\
\quad \quad \quad + \abs{f_{k}(\sigma^{k-1}(z)) - f_{k} \circ \sigma(\sigma^{k-1}(z))}
\end{array} \\
                                  & \leq \sqrt[\lambda]{2} k  {\color{blue}\bf \text{,}
      }\end{align*}
    where the last inequality follows from Proposition \ref{pointest}.
  \end{proof}

  We will now consider words of finite length on the alphabet $\{0,1\}$. For a given \(k \in \natural\), \(W_{k}\) is
  the set of words $w=[w_1,w_2,...,w_s]$, $s \leq k$, of length at most \(k\) and \(\hat{W}_{k} \cong
  \left\{0,1\right\}^{k}\) is the set of words $w=[w_1,w_2,...,w_s]$ of length exactly \(k\). By abuse of language we
  say that $\sigma( [x_1,x_2,..,x_k] )=[x_2,..,x_k]$, and words $[w_1,w_2,...,w_s]$ can also represent cylinder sets
  $[w_1,w_2,...,w_s]\subset \{0,1\}^\mathbb{N}$.

 Given $x=(x_1,x_2,\cdots,x_n,\cdots)$, we denote by $x|_{k}$ the word $[x_1,x_2,\cdots,x_k]$ of length $k$.

  It will be appropriate to define a metric \(d_{k}\) which is a graph distance in $\hat{W}_{k}$.

   \medskip

  \begin{proposition} \label{graphdist}
    Consider the graph \((V_{k},E_{k})\) given by:
    \begin{align*}
      V_{k} = \hat{W}_{k} \text{ and } E_{k} = \left\{(u,v) \in \hat{W}_{k} \times \hat{W}_{k} \mid \left(\sigma([u]) \cap [v]\right) \cup \left([u] \cap \sigma([v])\right) \neq \emptyset\right\}  {\color{blue}\bf \text{.}
    }\end{align*}
     Let \(d_{k}\) denote the graph distance in \(V_{k} = \hat{W}_{k}\). Then, 
    \begin{align*}
      \sup_{k} d_{k}(x|_{k},y|_{k}) \leq d_p(\delta_{x}, \delta_{y}) \leq \sqrt[\lambda]{2} \sup_{k} d_{k}(x|_{k},y|_{k})  {\color{blue}\bf \text{.}
    }\end{align*}
  \end{proposition}
  \begin{proof}
    Let \(P = \left((w_{1} = x|_{k}, w_{2}), (w_{2}, w_{3}), \cdots, (w_{n-1}, w_{n} = y|_{k})\right)\) be a path
    joining \(x|_{k}\) and \(y|_{k}\) along \(E_{k}\). Also, let \(f_{k}^{\theta}\) be a function that only depends
    on the first \(k\) coordinates. If \(f_{k}^{\theta}\) satisfies \(\norm{\left[\dirac_{p}, \pi(f_{k}^{\theta})\right]} \leq 1\),
    then by definition of \(E_{k}\), we have that \(\abs{\theta_{w_{i}} - \theta_{w_{i+1}}} \leq \sqrt[\lambda]{2}\). 
    Also notice the family of all \(f_{k}^{\theta}\) is dense in \(\cont(\Omega)\). From the above we get:
    \begin{align*}
      d_p(\delta_{x}, \delta_{y}) & = \sup_{\substack{f \in \cont(\Omega) \\ \norm{\left[\dirac_{p}, \pi(\mult_f)\right]} \leq 1}} \abs{f(x) - f(y)} \\
                                & = \sup_{k} \sup_{\substack{\theta \in \real^{\hat{W}_{k}} \\ \norm{\left[\dirac_{p}, \pi(f_{k}^{\theta})\right]} \leq 1}} \abs{f_{k}^{\theta}(x) - f_{k}^{\theta}(y)} \\
                                & = \sup_{k, \theta} \abs{f_{k}^{\theta}(x) - f_{k}^{\theta}(y)} \\
                                & = \sup_{k, \theta} \abs{\theta_{x|_{k}} - \theta_{y|_{k}}} \\
                                & \leq \sqrt[\lambda]{2} \sup_{k} d_{k}(x|_{k},y|_{k})  \text{.}
\end{align*}

    Now define \(f_{k}^{\gamma}\) by:
    \begin{align*}
      \gamma_{w} = d_{k}(w,y|_{k})  {\color{blue}\bf \text{.}
    }\end{align*}
    It is clear from the definition that:
    \begin{align*}
      \norm{\left[\dirac_{p}, \pi(f_{k}^{\theta})\right]} \leq 1 \text{ and } \abs{f_{k}^{\gamma}(x) - f_{k}^{\gamma}(y)} = d_{k}(x|_{k},y|_{k})  {\color{blue}\bf \text{,}
    }\end{align*}
    which gives the other inequality.
  \end{proof}
  \begin{remark}
    It can happen that:
    \begin{align*}
      \sup_{n \in \natural} \abs{\# \left\{i \leq n \mid x_{i} = 1\right\} - \# \left\{i \leq n \mid y_{i} = 1\right\}} < +\infty  {\color{blue}\bf \text{,}
    }\end{align*}
    and yet \(d_p(\delta{x}, \delta{y}) = +\infty\). For instance, consider the sequences:
    \begin{align*}
      x = (0,1,0,1,0,1,0,1,\cdots) \text{ and } y = (0,0,1,1,0,0,1,1,\cdots)  {\color{blue}\bf \text{.}
    }\end{align*}
    Their incidences of \(1\)'s up to length \(n\) differ by at most \(1 \ll +\infty\), and yet the distance between
    the truncations of this two points on the graph \((V_{k},E_{k})\) is of the same order of magnitude as \(k\).
    Then Proposition \ref{graphdist} gives \(d(\delta_{x},\delta_{y}) = +\infty\).
  \end{remark}

  \begin{proposition}
    Consider the length \(\ell\) of the longest common subword \(c\) of \(u\) and \(v\), which are words of length \(k\); that is,
    there exists \(m, n \in \natural\) such that:
    \begin{align*}
       \begin{array}{ccccccccc}
         c_{1} & = & u_{m} & = & v_{n} \text{,} \\
         c_{2} & = & u_{m+1} & = & v_{n+1} \text{,} \\
         c_{3} & = & u_{m+2} & = & v_{n+2} \text{,} \\
               & & \vdots &   & \\
      c_{\ell} & = & u_{m+\ell} & = & v_{n+\ell} \text{.}
      \end{array}
    \end{align*}

    Then the \(k^{\text{th}}\) graph distance from \(u \in \hat{W}_{k} \) to \(v\in \hat{W}_{k} \) satisfies
    \begin{align*}
        d_{k}(u,v) = \min \left\{k, k -\ell + m + 2n, k -\ell +n +2m\right\}  {\color{blue}\bf \text{.}
    }\end{align*}

  \end{proposition}
  \begin{proof}
    From the definition of \(E_{k}\) it follows that two vertices \(w, w^{\prime} \in V_k = \hat{W}_{} = \left\{0,1\right\}^{k}\)
    are ``connected'' if and only if they are of the form:
    \begin{align*}
      \begin{array}{ccc}
         w_{1} & = & w^{\prime}_{2} \\
         w_{2} & = & w^{\prime}_{3} \\
         w_{3} & = & w^{\prime}_{4} \\
               & \vdots & \\
         w_{k-1} & = & w^{\prime}_{k}
\end{array} & & \text{, or: } & &
      \begin{array}{ccc}
         w_{2} & = & w^{\prime}_{1} \\
         w_{3} & = & w^{\prime}_{2} \\
         w_{4} & = & w^{\prime}_{3} \\
               & \vdots & \\
         w_{k} & = & w^{\prime}_{k-1}
\end{array}  {\color{blue}\bf \text{.}
    }\end{align*}
    In particular, each vertex of \(w\) has four neighbours \(w^{\prime}\), uniquely characterized by one of the alternatives:
    \(w^{\prime}_{1} = 0\), \(w^{\prime}_{1} = 1\), \(w^{\prime}_{k} = 0\), or \(w^{\prime}_{k} = 1\).

    We say that \(c\) is \textit{on the same side} in \(u\) and \(v\) if \(m \lor n \leq \frac{\left(k - \ell\right)}{2}\)
    or if \(m \land n \geq \frac{\left(k - \ell\right)}{2}\). Otherwise, we say that \(c\) is \textit{on opposite sides}
    in \(u\) and \(v\).

    Of course, the diameter of the graph is \(k\). Now let us exhibit a shorter path from \(u\) to \(v\) when their
    maximal common word \(c\) has length \(\ell \geq \frac{k}{2}\) and is on the same side in \(u\) and \(v\). We
    are going to do the case \(m < n \leq \frac{\left(k - \ell\right)}{2}\); the other possibilities are analogous.
    \begin{enumerate}
      \item Starting at \(w^{1} = u\), take \(m-1\) steps to the neighbour \(w^{i+1}\) such that \(w^{i+1}_{k} = 0\).
      \item This will get us to \(w^{m} = [w, u|^{m+\ell+1}, 0^{m-1}]\).
      \item Take \(k-\ell-1\) steps to the neighbour \(w^{i+1}\) such that \(w^{i+1}_{1} = v_{n+m-i}\).
      \item This will get us to \(w^{m+k-\ell} = [v|_{n-1},w,u|^{k-\ell-n+1}_{k-n+m}]\).
      \item Take \(n-1\) steps to the neighbour \(w^{i+1}\) such that \(w^{i+1}_{1} = 0\).
      \item This will get us to \(w^{m+k-\ell+n} = [0^{k-\ell-n+1}, v|_{n-1}, w]\).
      \item Take \(n-1\) steps to the neighbour \(w^{i+1}\) such that \(w^{i+1}_{1} = v_{k - \ell + 2n + m - i}\).
      \item This will get us to \(w^{k-\ell+m+2n} = v\).
    \end{enumerate}
  \end{proof}
  \begin{remark}
    In particular, \(d_{k}(u,v) \geq k - \ell\), $u,v \in \hat{W}_{k}$.
  \end{remark}

  \begin{proposition} \label{xtildey}
    If \(x\) and \(y\) are two points such that \(d_p(\delta_{x},\delta_{y}) < +\infty\), then there exist \(m, n \in \natural\) such
    that \(\sigma^{m}(x) = \sigma^{n}(y)\).
  \end{proposition}
  \begin{proof}
    Let \(\ell(k)\) denote the size of the largest common word between \(x|_{k}\) and \(y|_{k}\). By Proposition
    \ref{graphdist} we have that \(d(\delta_{x},\delta_{y}) < +\infty \implies \sup_{k} k - \ell(k) < +\infty\). But
    this implies that there exists a \(k_{0}\) such that \(k - \ell(k) \equiv r = r(x,y)\) for every \(k \geq k_{0}\).
    This means that \(x\) and \(y\) only differ for finitely many terms, so there exist \(m, n \in \natural\) such
    that \(\sigma^{m}(x) = \sigma^{n}(y)\).
  \end{proof}

  By combining Proposition \ref{xtildey} with Example \ref{smxsny} we have just proved:
  \begin{theorem} \label{acima}
    For any pair of points \(x, y \in \Omega\), $p \geq 1$,
    \begin{align*}
        d_p(\delta_{x}, \delta_{y}) < + \infty \iff \text{ there exist } m, n \in \natural \text{ such that } \sigma^{m}(x) = \sigma^{n}(y)  {\color{blue}\bf \text{.}
    }\end{align*}
    \qed
  \end{theorem}

  The properties mentioned above in Theorem \ref{acima} are somehow related to the so-called {\it homoclinic equivalence
  relation} as defined in Section 6 in \cite{CLM}; the particular case where $\sigma^{m}(x) = \sigma^{m}(y)$, $m \in \mathbb{N}$ 
  (see also \cite{EL2}). In this case we get $d_p(\delta_{x}, \delta_{y}) < + \infty.$ For the case $m\neq n$ see
  Section 9 in \cite{ExVe} (and also \cite{KRe}).

  \begin{center} \label{fig:dists-12}
  \includegraphics[width=0.32\textwidth]{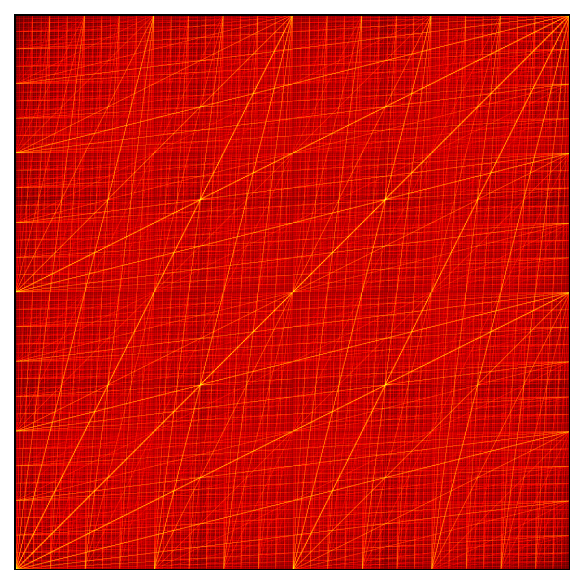}
  \includegraphics[width=0.32\textwidth]{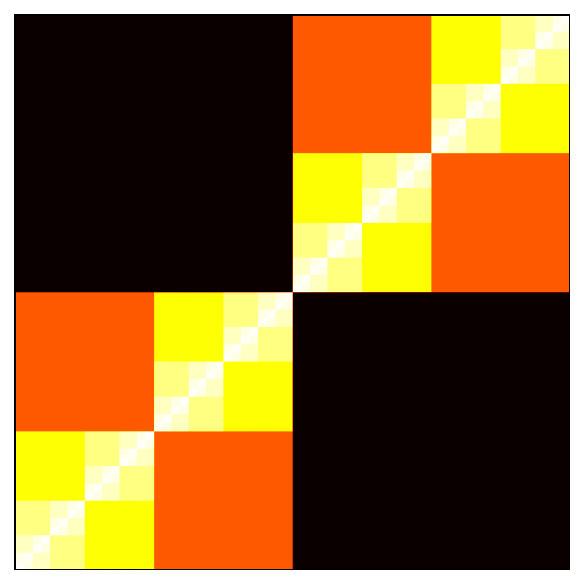}
  \includegraphics[width=0.32\textwidth]{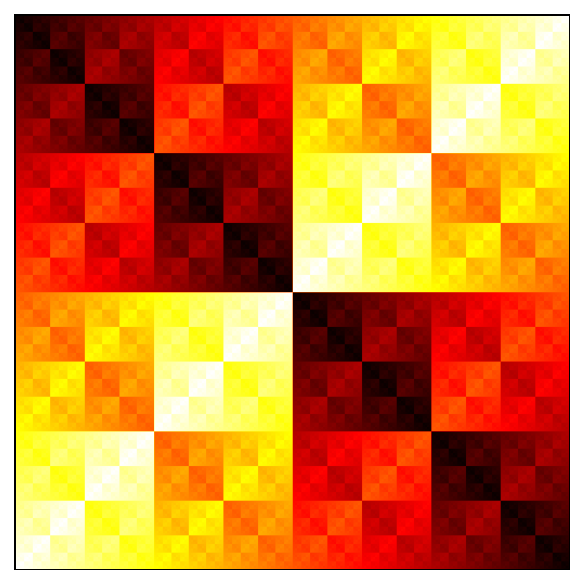}
  \end{center}
  \((V_{k}, E_{k})\)-graph, truncated and cumulative distances on \(\left\{0,1\right\}^{12}\). \linebreak Or
  \(d_{k}(u,v)\), \(2^{-N(u,v)}\) and \(\sum_{u_{i} = v_{i}} 2^{-i}\) respectively, \(N(u,v)\) being the smallest index
  for which \(u\) differs from \(v\). Here, we identified the word \(u \in \hat{W}_{k}\) with the real number \(\sum
  u_{i} 2^{-i} \in [0,1]\). The distance from \(u\) to \(v\) is plotted lighter if it is close to zero and darker if it
  is close to the diameter of \(\Omega\) (\(1\) for truncated and cumulative distances and \(k\) for \(d_{k}\)). This
  figure shows \(k = 12\).

  \section{General States\texorpdfstring{ - $d_p(\mu, \nu)$}{}} \label{gen}

  Now we pass to the question of computing and estimating the Connes distance $ d_p (\mu, \nu )$, for two general states
  \(\mu, \nu \in \states(\cont(\Omega)) = \prob(\Omega)\). First, we will exhibit an analog of Example \ref{ddxdsx}:

  \begin{proposition}
    If \(\sigma_{\sharp}\) denotes the push-forward through \(\sigma\), then for any given state \(\mu \in \prob(\Omega)\),
    Proposition \ref{pointest} implies: when $\lambda = \max\{p, p^\prime\}$
    \begin{align*}
      d_p(\mu, \sigma_{\sharp}(\mu)) & = \sup_{\substack{f \in \cont(\Omega) \\ \norm{\left[\dirac_{p}, \pi(\mult_f)\right]} \leq 1}} \abs{\int f \d \mu - \int f \d \sigma_{\sharp}(\mu)} \\
                                     & = \sup_{\substack{f \in \cont(\Omega) \\ \norm{\left[\dirac_{p}, \pi(\mult_f)\right]} \leq 1}} \abs{\int f \circ \sigma - f \d \mu} \\
                                     & \leq \sqrt[\lambda]{2}  {\color{blue}\bf \text{.}
    }\end{align*}
    \qed
  \end{proposition}

  For the next definitions the compact metric space $(X,\tilde{d})$ will represent either the set $\{0,1\}^\mathbb{N}$,
  or the set $V_k= \{0,1\}^k$, $k \geq 1.$

  \medskip

   Given a metric $\tilde{d}$ the Wasserstein distance between the probabilities $\mu$ and $ \nu$ on $X$ is (see \cite{Villa})

    \begin{equation} \label{wawa}
     W_{\tilde{d}}(\mu, \nu) = \sup_{\substack{f \in \cont(X)\\\abs{f(x)-f(y)}\leq \tilde{d}(x,y)}} \abs{\int f \d \mu - \int f \d \nu} \text{.}
   \end{equation}

   Next, an analog of Proposition \ref{graphdist} will consider the case where $\tilde{d}=d_k$ and $X=V_k$. Later we
   will introduce a metric $\tilde{d}=d^\infty$ on $X=\Omega$.

   \begin{proposition} \label{sgdist1}
     Let \(\mu, \nu \in \prob(\Omega)\) be any two states. Then, given $x,y\in \Omega=\{0,1\}^\mathbb{N}$, when $\lambda = \max\{p, p^\prime\}$
     \begin{align*}
         \sup_{\substack{k \in \natural \\ y \in \Omega}} \abs{\int d_{k} (x|_{k},y|_{k}) \d (\mu - \nu)} \leq d_p(\mu, \nu) \leq \sqrt[\lambda]{2} \sup_{\substack{k \in \natural \\ y \in \Omega}} \abs{\int d_{k} (x|_{k},y|_{k}) \d (\mu - \nu)}  {\color{blue}\bf \text{.}
     }\end{align*}
     Or:
     \begin{align*}
       \sup_{k \in \natural} \wass_{d_{k}}(\mu, \nu) \leq d_p(\mu, \nu) \leq \sqrt[\lambda]{2} \sup_{k \in \natural} \wass_{d_{k}}(\mu, \nu)  {\color{blue}\bf \text{,}
     }\end{align*}
   \end{proposition}
   \begin{proof}
     Analogous to the proof of Theorem \ref{conneswass}. Observe the Wasserstein distance is equal to the supremum
     in \(y\) for each respective \(k\). It is also increasing in \(k\) so that the suprema are actually limits.
   \end{proof}

   Finally, we have that:
   \begin{theorem} \label{conneswass}
     Let \(\mu, \nu \in \prob(\Omega)\) be any two states. Then, for $p\geq 1$, when $\lambda = \max\{p, p^\prime\}$
     \begin{align*}
         \wass_{d^{\infty}}(\mu, \nu) \leq d_p(\mu, \nu) \leq \sqrt[\lambda]{2} \wass_{d^{\infty}}(\mu, \nu)  {\color{blue}\bf \text{,}
     }\end{align*}
     where \(d^{\infty}\) is given by:
     \begin{align*}
         d^{\infty} (x, y) \defn \min_{\substack{m, n \in \natural \\ \sigma^{m}(x) = \sigma^{n}(y)}} m + n  {\color{blue}\bf \text{.}
     }\end{align*}
   \end{theorem}
   \begin{proof}
     Consider a function \(f \in \cont(\Omega)\) such that  \(\norm{\left[ \dirac_{p}, \pi(\mult_f) \right]} \leq 1\).
     Proposition \ref{pointest} shows that \(\abs{f \circ \sigma(x) - f(x)} \leq \sqrt[\lambda]{2} \). Now, if \(m,
     n \in \natural\) are such that \(\sigma^{m}(x) = \sigma^{n}(y)\), then:
     \begin{align*}
       \abs{f(x) - f(y)} & = \abs{\begin{array}{l}
         f(x) - f \circ \sigma(x) + \\
       \quad + f \circ \sigma(x) - f(y) \end{array}} \\
                         & = \abs{\begin{array}{l}
         f(x) - f \circ \sigma(x) + \\
         \quad + f \circ \sigma(x) - f \circ \sigma^{2}(x) + \\
       \quad \quad + f \circ \sigma^{2}(x) - f(y) \end{array}} \\
                         & = \abs{\begin{array}{l}
         f(x) - f \circ \sigma(x) + \\
\quad + \cdots + \\
\quad \quad + f \circ \sigma^{m-1}(x) - f \circ \sigma^{m}(x) + \\
\quad \quad \quad + f \circ \sigma^{n}(y) - f \circ \sigma^{n-1}(y) + \\
\quad \quad \quad \quad + \cdots + \\
\quad \quad \quad \quad \quad + f \circ \sigma(y) - f (y) \\
\end{array}} \\
                         & \leq \sqrt[\lambda]{2} \left( m + n \right)  {\color{blue}\bf \text{,}
     }\end{align*}
     which shows that \(\abs{f(x) - f(y)} \leq \sqrt[\lambda]{2} d_{\infty}(x, y)\). On the other hand, if \(f \in \cont(\Omega)\)
     is such that \(\abs{f(x) - f(y)} \leq d_{\infty}(x, y)\), then \(\abs{f \circ \sigma(x) - f(x)} \leq 1\).
     Therefore:
     \begin{align*}
        \norm{\left[\dirac_{p}, \pi(\mult_f)\right]} & = \sqrt[\lambda]{
    \begin{array}{l}
    \frac{\abs{f(x) - f(0x)}^{\lambda}}{2} + \\ \quad + \frac{\abs{f(x) - f(1x)}^{\lambda}}{2}
    \end{array}} \\
                                                & \leq \sqrt[\lambda]{ \frac{1}{2} + \frac{1}{2} } \\
                                                & = 1  {\color{blue}\bf \text{.}
     }\end{align*}
  \end{proof}

  \smallskip

  Note that the distance $ d^{\infty} $ does not produce the same topology as the one obtained from the usual metric on $\Omega$.

  \smallskip

  \begin{corollary} \label{oprt}
    The Connes distance between two states \(\mu, \nu \in \prob(\Omega)\) is finite if and only if they give the same
    weight to each equivalence class of the relation given by \(x \rel y \iff \exists m, n \in \natural
    : \sigma^{m}(x) = \sigma^{n}(y)\), that is: if \(\mu(\bar{x}) = \nu(\bar{x})\) for any \(x \in \Omega\),
    \(\bar{x}\) the equivalence class of \(x\). Each of these equivalence classes is the connected component of each
    of its elements with respect to distance \(d^{\infty}\).
  \end{corollary}

  Taking the limits \(p\) in Theorem \ref{conneswass} gives \(d_{\dirac_{1}} = d_{\dirac_{\infty}} = W_{d^{\infty}}\).
  Notice, \(d_ {\dirac_{p}} = \sqrt[\lambda]{2} d_{\infty}\). This shows the family \(d_{\dirac_{p}}\) interpolates
  between the Connes and Wasserstein distances; the Wasserstein distance corresponds to the cases \(p = 1, +\infty\). 
  We have numerical evidence for the inequalities in Theorem \ref{conneswass} being strict.

  \newpage

  \section{Appendix} \label{kntrvch}
The material of this section was taken from \cite{diss}.
    The noncommutative generalization of the optimal transport problem  in \cite{diss} is a version 
  of the Monge-Kantorovich optimal transport problem (on compact spaces with positive continuous cost) according to the 
  following dictionary: \newline \newline
   \resizebox{\textwidth}{!}{
    \begin{tabular}{|c|c|}
    \hline
          & \\
    Real functions & Self-adjoint elements \\
    \(f \in \cont(X)\), \(g \in \cont(Y)\) & \(a \in \mathcal{A}\), \(b \in \mathcal{B}\) \\
          & \\
    \hline
          & \\
    Probability measures & C*-algebra states \\
    \(\mu \in \prob(X)\), \(\nu \in \prob(Y)\) & \(\eta \in \states(\mathcal{A})\), \(\xi \in \states(\mathcal{B})\) \\
          & \\
    \hline
          & \\
    Cost function & Cost element \\
    \(c \in \cont(X \times Y)\), \(c \geq 0\) & \(c \in \mathcal{A} \otimes \mathcal{B}\), \(c \geq 0\) \\
          & \\
    \hline
          & \\
    Coupling probabilities & Coupling states \\
    \(\rho \in \prob(X \times Y)\) & \(\omega \in \states(\mathcal{A} \otimes \mathcal{B})\) \\
          \(\int f(x) + g(y) \d \rho = \int f(x) \d \mu + \int g(y) \d \nu\) & 
          \(\omega(a \otimes 1 + 1 \otimes b) = \eta(a) + \xi(b)\) \\
          & \\
    \hline
          & \\
          \(W_{c}(\mu, \nu) \defn \inf_{\rho} \int c \d \rho\) & \(\mathcal{W}_{c}(\eta, \xi) \defn \inf_{\omega}
          \omega(c)\) \\
          & \\
    \hline
    \end{tabular}} \newline
    \vspace{5pt}
    
    \bigskip
    
 Therefore,  it is natural to pursue the following reasoning: 

     Let \(\mathcal{A}\) and \(\mathcal{B}\) be two unital C*-algebras.

     Denote by \(\mathcal{A} \otimes \mathcal{B}\) the maximal tensor product between \(\mathcal{A}\) and \(\mathcal{B}\).

    Let \(c \in (\mathcal{A} \otimes \mathcal{B})^{+}\) be a positive element (henceforth called \textit{cost element}).

     Let \(\eta \in \states(\mathcal{A})\) and \(\xi \in \states(\mathcal{B})\) be two given C*-algebra states.

     Denote by \(\Gamma(\eta, \xi)\) the set of all states \(\omega \in \states(\mathcal{A} \otimes \mathcal{B})\) such that:
    \begin{equation*}
     \omega(a \otimes 1_{\mathcal{B}} + 1_{\mathcal{A}} \otimes b) = \eta(a) + \xi(b) 
    \quad \forall \, a \in \mathcal{A}, b \in \mathcal{B} \text{.}
    \end{equation*}
     These states are called the \textit{admissible couplings of \(\eta\) and \(\xi\)}.

     Denote by \(\tilde\Gamma(c)\) the set of all pairs of self-adjoint elements \((a,b) \in \mathcal{A} \times
    \mathcal{B}\) such that:
    \begin{equation*}
    a \otimes 1_{\mathcal{B}} + 1_{\mathcal{A}} \otimes b \leq c 
    \end{equation*}
    These pairs are called the \textit{admissible pairings for the cost \(c\)}.

     The noncommutative optimal transport problem from \cite{diss} consists of minimizing the evaluation of the cost
    element among all admissible couplings of \(\eta\) and \(\xi\). We call this value the minimum  optimal transport cost from
    \(\eta\) to \(\xi\). A value \(c\) is fixed according to convenience in each  problem. Thus, we write:
    \begin{equation} \label{ncopt}
      \mathcal{W}_{c}(\eta, \xi) \defn \inf_{\omega \in \Gamma(\eta, \xi)} \omega(c)
    \end{equation}

     It is possible to prove that minimizers for \eqref{ncopt} exist by employing the direct method of the calculus of
    variations (see \cite{diss}).  Furthermore, and also done in \cite{diss}, it is possible to prove
    a formula analogous to the Kantorovich duality formula for \eqref{ncopt}, as we will see. 
    Notice this recovers the existence of minimizers as a corollary.

    \begin{theorem}
       Let \(\mathcal{A}\), \(\mathcal{B}\), \(c \in (\mathcal{A} \otimes \mathcal{B})^{+}\), \(\eta \in \states(\mathcal{A})\),
      and \(\xi \in \states(\mathcal{B})\) be as above, and consider the aforementioned definitions of \(\Gamma(\eta, \xi)\) and
      \(\tilde\Gamma(c)\). 
      Then:
    \begin{align*}
       \mathcal{W}_{c}(\eta, \xi) \defn \inf_{\omega \in \Gamma(\eta, \xi)} \omega(c) = \sup_{\substack{a \in A\\b \in \mathcal{B}\\a \otimes 1 + 1 \otimes b \leq
      c}} \eta(a) + \xi(b) \text{.}
    \end{align*}
    \end{theorem}
    \begin{proof}
      We very closely follow \cite{Villa}, which amounts to employing the Fenchel-Rockafellar duality theorem. In the
      notation therein, our normed vector space \(E\) is the real vector space of self-adjoint elements of \(\mathcal{A}
      \otimes \mathcal{B}\), and our convex functions \(\Theta : E \to \real\) and \(\Xi : E \to \real\) are given by:
      \begin{align*}
         \Theta (x) & \defn \left\{
          \begin{aligned}
            0 	\quad	& \text{ if } x \geq -c \text{,} \\
            + \infty \quad & \text{ otherwise.}
        \end{aligned} \right. \intertext{ And:}
         \Xi (x) & \bf \defn \left\{
          \begin{aligned}
            \eta (a) + \xi (b) \quad &	\text{ if } x = a \otimes 1_{\mathcal{B}} + 1_{\mathcal{A}} \otimes b \text{,} \\
            + \infty \quad & \text{ otherwise.}
          \end{aligned} \right.
      \end{align*}

      The point \(x_{0} = 1_{\mathcal{A} \otimes \mathcal{B}}\) lies in the intersection of the effective domains of
      both functions (that is, \(\Theta(1_{\mathcal{A} \otimes \mathcal{B}}) < + \infty\) 
      and \(\Xi(1_{\mathcal{A} \otimes \mathcal{B}}) < + \infty\)), because \(1_{\mathcal{A} \otimes \mathcal{B}} \geq
      0 \geq -c\) and:
      \begin{align*}
        \Xi (1_{\mathcal{A} \otimes \mathcal{B}}) &= \Xi (1_{\mathcal{A}} \otimes 1_{\mathcal{B}}) \\
                  &= \Xi \left(\frac{1}{2} 1_{\mathcal{A}} \otimes 1_{\mathcal{B}} + 1_{\mathcal{A}} \otimes \frac{1}{2} 1_{\mathcal{B}}\right) \\
                  &=\eta \left(\frac{1}{2} 1_{\mathcal{A}}\right) + \xi \left(\frac{1}{2} 1_{\mathcal{B}}\right) \\
                  &= \frac{1}{2} + \frac{1}{2} \\
                  &= 1 \text{.}
      \end{align*}

      The function \(\Theta\) is continuous at the point \(x_{0} = 1_{\mathcal{A} \otimes \mathcal{B}}\) because this
      point lies in the interior of its effective domain. For example, the set \(\mathcal{A} \otimes \mathcal{B}^{+^{\circ}}\) 
      of all strictly positive elements of \(\mathcal{A} \otimes \mathcal{B}\) is an open set entirely contained in the
      effective domain; and \(1_{\mathcal{A} \otimes \mathcal{B}}\) pertains to such set.

      Applying the Fenchel-Rockafellar duality, we conclude:
      \begin{equation} \label{eq:frd}
        \inf_{x \in E} \Theta(x) + \Xi(x) = \max_{\chi \in E^{\star}} - \Theta^{\star} ( - \chi ) - \Xi^{\star} ( \chi )
        \text{.}
      \end{equation}

      Now we pass to the issue of computing both sides of \eqref{eq:frd}. On the left side, we have:
      \begin{align*}
        \inf_{ \substack{ x = a \otimes 1_{\mathcal{B}} + 1_{\mathcal{A}} \otimes b \\ x \geq -c }} \eta (a) + \xi (b)
        & = - \sup_{ \substack{ x = a \otimes 1_{\mathcal{B}} + 1_{\mathcal{A}} \otimes b \\ x \leq c }} \eta (a) + \xi (b) \\
        & = - \sup_{ \left( a, b \right) \in \tilde\Gamma(c) } \eta (a) + \xi (b)
        \text{.}
      \end{align*}

      On the right side, the Legendre transform of \(\Theta\) is:
      \[
        \Theta^{\star} ( - \chi )
        = \underset{x \in E}{\sup} \, - \chi x - \Theta x
        = \underset{x \geq -c}{\sup} - \chi x
        \text{.}
      \]
      If $\chi \ngeq 0$, there must be some $x \geq 0 \geq -c$ for which $\chi(x) < 0$. Given such $x$, the family of
      positive elements $nx$ ensures that the supremum be $+ \infty$. If otherwise $\chi \geq 0$, then we can compare
      the evaluation of \(\chi\) at $-c$ with the evaluation of \(\chi\) at any other $x$:
      \[
        - \chi (-c) - \left[ - \chi \left(x\right) \right]
        = - \chi (-c -x)
        = \chi (c + x)
        \geq 0
        \text{,}
      \]
      and see that the supremum must be $\chi(c)$.

      Still on the right side of \eqref{eq:frd}, the Legendre transform of $\Xi$ is:
      \begin{align*}
        \Xi^{\star} ( \chi )
        & = \sup_{x \in E} \, \chi x - \Xi x \\
        & = \sup_{x = a \otimes 1_{\mathcal{B}} + 1_{\mathcal{A}} \otimes b} \chi ( a \otimes 1_{\mathcal{B}} + 1_{\mathcal{A}} \otimes b ) - ( \eta (a) + \xi (b) )
        \text{.}
      \end{align*}
      If, for any of these $x$, the quantity $\chi ( a \otimes 1_{\mathcal{B}} + 1_{\mathcal{A}} \otimes b ) - ( \eta (a) + \xi (b) )$ 
      is not zero, then either the family $nx$ or $-nx$ ensures the supremum be $+ \infty$. In the absence of such $x$,
      $\Xi^{\star} ( \chi ) = 0$. Synthetically:
      \begin{align*}
        \Theta^{\star} ( - \chi ) & = \left\{
          \begin{aligned}
            \chi(c) 			& \text{ if } \chi \geq 0 \text{,} \\
            + \infty 			& \text{ otherwise.}
        \end{aligned} \right. \intertext{And:}
        \Xi^{\star} ( \chi ) & = \left\{
          \begin{aligned}
            0	 	& \text{ if } \chi ( a \otimes 1_{\mathcal{B}} + 1_{\mathcal{A}} \otimes b ) = \eta (a) + \xi (b) \text{,} \\
            + \infty	& \text{ otherwise.}
          \end{aligned} \right.
      \end{align*}
      Conveniently, the intersection of the effective domains of such functions is precisely $\Gamma(\eta,\xi)$.

      Finally, we rewrite Fenchel-Rockafellar duality in terms of the previous observations:
      \begin{align*}
        \inf_{a \otimes 1_{\mathcal{B}} + 1_{\mathcal{A}} \otimes b \geq -c} \eta (a) + \xi (b)
        & = \max_{\substack{\chi \geq 0 \\ \chi ( a \otimes 1_{\mathcal{B}} + 1_{\mathcal{A}} \otimes b ) = \eta (a) + \xi (b)}} - \chi(c)
        \text{,}
      \intertext{and exchange signs, to obtain:}
        \sup_{\left( a, b \right) \in \tilde\Gamma(c)} \eta (a) + \xi (b)
        & = \min_{\omega \in \Gamma(\eta, \xi)} \omega(c)
        \text{.}
      \end{align*}
    \end{proof}

    When \(\mathcal{A} = \mathcal{B} = \cont(X)\) and the cost \(c = d\) is a metric, we recover the following
    form of the Kantorovich duality formula:
    \begin{align*}
      \mathcal{W}_{d}(\mu, \nu) = \sup_{\substack{f \in \cont(X)\\\abs{f(x) - f(y)} \leq d(x,y)}} \abs{\int f \d \mu - \int f \d
      \nu} \text{.}
    \end{align*}

     \bigskip

    This work is part of the PhD thesis of William M. M. Braucks in Programa de   P\'os-gradua\c c\~ao (see \cite{rkboson}) em   Matem\'atica - UFRGS (2025){\color{blue}\bf .
}

    \bigskip

    William M. M. Braucks (braucks.w@gmail.com)
    \smallskip

     Artur O. Lopes (arturoscar.lopes@gmail.com)
    \smallskip

    Inst. Mat. Est. - UFRGS - Porto Alegre, Brazil


\begin{thebibliography}{99}

\bibitem{tentropy}
A.~Antonevich, V.~Bakhtin, and A.~Lebedev.
\newblock On t-entropy and variational principle for the spectral radii of transfer and weighted shift operators.
\newblock {\em Ergodic Theory and Dynamical Systems}, 31(4):995--1042, 2011.

\bibitem{lpcuntz}
K.~Bardadyn.
\newblock {$ L^{p} $}-cuntz algebras and spectrum of weighted composition operators.
\newblock In {\em Workshop on Geometric Methods in Physics}, pages 189--198. Springer, 2022.

\bibitem{specweightiso}
K.~Bardadyn and B.~K. Kwa{\'s}niewski.
\newblock Spectrum of weighted isometries: {C*-algebras}, transfer operators and topological pressure.
\newblock {\em Israel Journal of Mathematics}, 246:149--210, 2021.

\bibitem{diss}
W.~M.~M. Braucks.
\newblock Transporte \'Otimo em {C*-algebras}.
\newblock 2021.

\bibitem{rkboson}
W.~M.~M. Braucks.
\newblock The {Dirac} operator for the {Ruelle-Koopman} pair: an interplay between {Connes} distance and symbolic dynamics. PhD thesis, P\'os-gradua\c c\~ao (see \cite{rkboson}) em   Matem\'atica - UFRGS 
\newblock 2025.

\bibitem{BL}
W.~M.~M. Braucks and A. O. Lopes. 
The Dirac operator for the pair of  Ruelle and Koopman operators and a generalized Boson formalism, ArXiv (2025)

\bibitem{CHLS}
L.~Cioletti, L.~Y. Hataishi, A.~O. Lopes, and M.~Stadlbauer.
\newblock Spectral triples on thermodynamic formalism and {Dixmier} trace representations of {Gibbs}: theory and examples.
\newblock {\em Nonlinearity, Volume 37, Number 10, 105010 (56pp)}, (2024).

\bibitem{cfrconnes}
A.~Connes.
\newblock Compact metric spaces, {Fredholm} modules, and hyperfiniteness.
\newblock {\em Ergodic theory and dynamical systems}, 9(2):207--220, 1989.

\bibitem{CLM}
G.~G. de~Castro, A.~O. Lopes, and G.~Mantovani.
\newblock {Haar} systems, {KMS} states on von {Neumann} algebras and {C*}-algebras on dynamically defined groupoids and noncommutative integration.
\newblock In {\em Modeling, Dynamics, Optimization and Bioeconomics IV: DGS VI JOLATE, Madrid, Spain, May 2018, and ICABR, Berkeley, USA, May--June 2017 - Selected \, Contributions}, pages 79--137. Springer, 2021.

\bibitem{lpspec}
A.~Delf{\'\i}n, C.~Farsi, and J.~Packer.
\newblock {$ L^{p} $}-spectral triples and $ p $-quantum compact metric spaces.
\newblock {\em arXiv preprint arXiv:2411.13735}, 2024.

\bibitem{cdixmier}
J.~Dixmier.
\newblock {C*-Algebras}, translated from the {French} by {Francis Jellett}.
\newblock {\em North-Holland Math. Library}, 15, 1977.

\bibitem{EL2}
R.~Exel and A.~Lopes.
\newblock {C*-algebras}, approximately proper equivalence relations and thermodynamic formalism.
\newblock {\em Ergodic Theory and Dynamical Systems}, 24(4):1051--1082, 2004.

\bibitem{ExVe}
R.~Exel and A.~Vershik.
\newblock {C*-Algebras} of irreversible dynamical systems.
\newblock {\em Canadian Journal of Mathematics}, 58(1):39--63, 2006.

\bibitem{Kesse}
 M. Kessebohmer and T. Samuel, Spectral metric spaces for Gibbs measures. J. Funct. Anal., 265(9):1801--1828, 2013


\bibitem{KRe}
A.~Kumjian and J.~Renault.
\newblock {KMS} states on {C*-algebras} associated to expansive maps.
\newblock {\em Proceedings of the American Mathematical Society}, 134(7):2067--2078, 2006.

\bibitem{connes-opt}
P.~Martinetti.
\newblock {Connes} distance and optimal transport.
\newblock In {\em Journal of Physics: Conference Series}, volume 968, page 012007. IOP Publishing, 2018.

\bibitem{PP}
W.~Parry and M.~Pollicott.
\newblock Zeta functions and the periodic orbit structure of hyperbolic dynamics.
\newblock {\em Ast{\'e}risque}, 187(188):1--268, 1990.

\bibitem{state-metrics}
M.~A. Rieffel.
\newblock Metrics on states from actions of compact groups.
\newblock {\em Documenta Mathematica}, 3:215--230, 1998.

\bibitem{Sharp1}
R.  Sharp,  Spectral triples and Gibbs measures for expanding maps on Cantor sets. J. Noncommut. Geom., 6(4):801-817,
 2012.
 
 \bibitem{Sharp2}
R. Sharp, Conformal Markov systems,Patterson-Sullivan measure on limit sets and spectral triples. Discrete Contin. Dyn.
 Syst., 36(5):2711–2727, 2016.

\bibitem{Villa}
C.~Villani.
\newblock {\em Topics in optimal transportation}, volume~58.
\newblock American Mathematical Soc., 2021.

\end{thebibliography}
\end{document}